\documentclass[twoside,a4paper]{article}

\usepackage{amsmath,graphicx,amssymb,fancyhdr,amsthm,,textcomp,hyperref}
\usepackage{float}
\usepackage{subcaption}
\usepackage[english]{babel}
\usepackage[numbers]{natbib}

\newtheorem{thm}{Theorem}[section]
\newtheorem{cor}[thm]{Corollary}
\newtheorem{lem}[thm]{Lemma}
\newtheorem{prop}[thm]{Proposition}
\newtheorem{con}[thm]{Conjecture}
\theoremstyle{definition}
\newtheorem{defn}[thm]{Definition}	
\theoremstyle{remark}

\usepackage{color}

\def\beq{\begin{eqnarray}}  
\def\eeq{\end{eqnarray}}  

\def\bsp{\begin{split}}  
\def\esp{\end{split}}

\newcommand{\mbold}[1]{\mbox{\boldmath{\ensuremath{#1}}}}

\def \bl {\mbox{{\mbold\ell}}}
\def \bn {\mbox{{\bf n}}}
\def \bm {\mbox{{\bf m}}}

\def \bT {\mbox{{\bf T}}}
\def \bX {\mbox{{\bf X}}}
\def \bY {\mbox{{\bf Y}}}
\def \bZ {\mbox{{\bf Z}}}
\def \bW {\mbox{{\mbold W}}}

\def\beq{\begin{eqnarray}}
\def\eeq{\end{eqnarray}}

\begin{document}

\title{$\mathcal{I}$-Preserving Diffeomorphisms of Lorentzian Manifolds}
\author{{\large\textbf{David Duncan McNutt$^{1}$  and Matthew Terje Aadne$^{1}$ }}
%EndAName  
%\address{  
\vspace{0.3cm} \\ 
$^{1}$ Faculty of Science and Technology,\\
University of Stavanger, 
N-4036 Stavanger, Norway  \\
\vspace{0.3cm} \\
\texttt{david.d.mcnutt@uis.no,matthew.t.aadne@uis.no }}  
\date{\today}  
\maketitle  
\pagestyle{fancy}  
\fancyhead{} % clear all header fields  
\fancyhead[EC]{}  
\fancyhead[EL,OR]{\thepage}  
\fancyhead[OC]{}  
\fancyfoot{} % clear all footer fields  

\begin{abstract}

We examine the existence of one parameter groups of diffeomorphisms whose infinitesimal generators annihilate all scalar polynomial curvature invariants through the application of the Lie derivative, known as $\mathcal{I}$-preserving diffeomorphisms. Such mappings are a generalization of isometries and appear to be related to nil-Killing vector fields, for which the associated Lie derivative of the metric yields a nilpotent rank two tensor. We show that the set of nil-Killing vector fields contains Lie algebras, although the Lie algebras may be infinite and can contain elements which are not $\mathcal{I}$-preserving diffeomorphisms. We then study  the curvature structure of a general Lorenztian manifold, or spacetime to show that $\mathcal{I}$-preserving diffeomorphism will only exists for $\mathcal{I}$-degenerate spacetimes and to determine when the $\mathcal{I}$-preserving diffeomorphisms are generated by nil-Killing vector fields. We identify necessary and sufficient conditions for the degenerate Kundt spacetimes to admit an additional $\mathcal{I}$-preserving diffeomorphism and conclude with an application to the class of Kundt spacetimes with constant scalar polynomial curvature invariants to show that a finite transitive Lie algebra of nil-Killing vector fields always exists for these spacetimes. 

\end{abstract}

\maketitle

\begin{section}{Introduction}
Unlike the Riemannian spaces where the set,  $\mathcal{I}$, of all scalar polynomial curvature invariants ($SPIs$): 
$$ \mathcal{I} = \{ R, R_{abcd} R^{abcd}, \ldots, R_{abcd;e} R^{abcd;e}, \ldots \}, $$
locally characterize the manifold completely, for the pseudo-Riemannian spaces there exists classes of manifolds which cannot be uniquely characterized locally by their $SPIs$. That is, for any such metric, ${\bf g}$, there exists a smooth (one parameter) deformation of the metric, ${\bf \tilde{g}}_{\tau}$, with ${\bf \tilde{g}}_{0} = {\bf g}$ and ${\bf \tilde{g}}_{\tau}$, $\tau >0$ not diffeomorphic to ${\bf g}$ yielding the same set $\mathcal{I}$, such a space is called {\it $\mathcal{I}$-degenerate} \cite{CSI4a, SHY2015}. 

In the case of a spacetime, i.e., a Lorentzian manifold, $(M, {\bf g})$, a more practical definition of $\mathcal{I}$-degeneracy can be stated in terms of the structure of the curvature tensor and its covariant derivatives. To discuss this, we must examine the effect of a boost on the null coframe $\{ \bn, \bl, \bm^i\}$, $\bl' = \lambda \bl,~~ {\bf n'} = \lambda^{-1} \bn,$ for which the components of an
arbitrary tensor, ${\bf T}$, of rank $n$ transform as
\beq T'_{a_1 a_2...a_n} = \lambda^{b_{a_1 a_2 ... a_n}} T_{a_1 a_2 ... a_n},~~ b_{a_1 a_2...a_n} = \sum_{i=1}^n(\delta_{a_i 0} - \delta_{a_i 1}),  \eeq
\noindent where $\delta_{ab}$ denotes the Kronecker delta symbol. The quantity, $b_{a_1 a_2 ... a_n}$, is called the {\it boost weight} (b.w) of the frame component $T_{a_1 a_2 ... a_p}$. Any tensor can be decomposed in terms of the b.w. of its components and this b.w. decomposition gives rise to the {\it alignment
  classification}, by identifying null directions relative to
which the components of a given tensor have a particular b.w. configuration. This classification reproduces the Petrov and Segre classifications in 4D, and also leads to a coarse classification in higher dimensions \cite{classa,classb,classc, OrtaggioPravdaPravdova:2013}. 

We will define the maximum b.w. of a tensor, ${\bf T}$, for a null direction $\bl$ as the boost order, and denote it as $\mathcal{B}_{{\bf T}}(\bl)$. The Weyl tensor and any rank two tensor, {\bf T}, can be broadly classified into five {\it alignment types}: {\bf G},{\bf I}, {\bf II}, {\bf III}, and {\bf N} if there exists an $\bl$ such that $\mathcal{B}_{{\bf T}} (\bl) = 2, 1, 0,-1,-2$ and we will say $\bl$ is ${\bf T}$-aligned, while if {\bf T} vanishes, then it belongs to alignment type {\bf O}. For higher rank tensors, like the covariant derivatives of the curvature tensor, the alignment types are still applicable despite the possibility that $|\mathcal{B}_{{\bf T}} (\bl)|$ may be greater than two. Any $\mathcal{I}$-degenerate spacetime admits a null frame such that all of the positive b.w. terms of the curvature tensor and its covariant derivatives are zero in this common frame, that is they are all of alignment type {\bf II}.

A significant subset of the $\mathcal{I}$-degenerate spacetimes are contained in a subclass of the Kundt spacetimes, for which the curvature tensors and its covariant derivatives are of alignment type {\bf II}, known as the {\it degenerate Kundt spacetimes}. In the three-dimensional (3D) and four-dimensional (4D) cases, all such spacetimes are contained in the degenerate Kundt spacetimes \cite{CHPP2009}. It is conjectured that any $D$-dimensional $\mathcal{I}$-degenerate spacetime is a degenerate Kundt spacetime \cite{CHP2010}. 

Of particular interest are those spacetimes where all elements of $\mathcal{I}$ vanish or are constant, such spacetimes are known as {\it vanishing scalar invariant} ($VSI$) or  {\it constant scalar invariant} ($CSI$) spacetimes respectively \cite{CSI4d}.
The class of CSI spacetimes are applicable to many theories of gravity, as they contain a subset of spacetimes that are universal, and hence solve the vacuum equations of all gravitational theories with a Lagrangian constructed from SPIs \cite{coleyhervik2011, hervik2015}.

In 3D and 4D, it has been shown that all $CSI$ spacetimes are either locally homogeneous or they belong to the degenerate Kundt class \cite{CSI4c, CSI4b}, while in higher dimensions it is conjectured that a $CSI$ spacetime will either be locally homogeneous or belong to the degenerate Kundt class \cite{CSI4d}. It has been shown that the $VSI$ spacetimes belong to the degenerate Kundt class in all dimensions \cite{Higher}. The subset of $CSI$ spacetimes belonging to the Kundt class are called {\it Kundt-CSI}. For Kundt-$CSI$ metrics, the transverse space is a locally homogeneous Riemannian manifold and the metric functions must satisfy particular differential equations \cite{CSI4c, CSI4b, CFH, SM2018}. 

In the Riemannian case, a space is $CSI$ if and only if the space is locally homogeneous. In the Lorentzian case there are Kundt-$CSI$ spacetimes that do not have enough Killing vector fields to determine the $CSI$ property. However, any Kundt-$CSI$ spacetime can be mapped to a related locally homogeneous Kundt-$CSI$ spacetime with the same set $\mathcal{I}$ which provides an explanation for the $CSI$ property \cite{SHY2015,SM2018}. Such a metric is known as a Kundt$^\infty$ triple and will be defined in section \ref{sec:KundtCSI}.  

The pseudo-Riemannian case admits $CSI$ metrics, with two known classes of metrics containing $CSI$ solutions, namely the Kundt and Walker pseudo-Riemannian metrics \cite{SHY2015}. Unlike the Lorentzian case, there exists $CSI$ pseudo-Riemannian spaces which are mapped to simpler $CSI$ spaces lacking a sufficient number of Killing vector fields required to prove the metrics are $CSI$. In such cases, all possible $SPIs$ up to an appropriate order must be checked explicitly to prove the $CSI$ property. As an example, consider the following neutral signature metric in $4D$:

\beq ds^2 = 2du (V du + dv) + dU (av^4 dU + dV), \label{neutralCSI} \eeq

\noindent where $a$ is a constant. Any $SPI$ constructed from the curvature tensor and its covariant derivative up to order $3$ all vanish, while all $SPIs$ constructed from the covariant derivatives of the curvature tensor of order $p \geq 3$ are constant. 

While this spacetime does not admit a sufficient number of Killing vector fields, it does admit a transitive set of vector fields, $$\left\{ \frac{\partial}{\partial u},~\frac{\partial}{\partial v},~\frac{\partial}{\partial U},~\frac{\partial}{\partial V} \right\}. $$
\noindent For each of these vector fields, the Lie derivative of the metric in the direction of the vector field  produces a nilpotent rank 2 tensor, that is, they are {\it nil-Killing} vector fields \cite{SH2018}. A subset of the nil-Killing vector fields known as Kerr-Schild vector fields have been studied as generators for Kerr-Schild transformations of spacetimes \cite{SC2000}. The Kerr-Schild vector fields have also been used to establish the existence of trapping horizons in 4D spacetimes \cite{Senov2003}. Generally the Kerr-schild vector fields are finite dimensional. However, in some cases the Kerr-Schild vector fields can form an infinite dimensional Lie algebra. 

In comparison, the four nil-Killing vector fields of the line-element \eqref{neutralCSI} form a finite abelian Lie algebra and the flows of each of the vector fields leave the elements of $\mathcal{I}$ invariant. Such a vector field generalizes the concept of an isometry by preserving SPIs without necessarily being an isometry of the metric, and so the corresponding flow of such a vector field is called an {\it $\mathcal{I}$-preserving diffeomorphism} ($IPD$). The associated vector fields of the $IPDs$ can help determine if a spacetime is $CSI$ without explicitly checking all SPIs \cite{MTA2018}. Motivated by this example, it is of interest to determine a simple criteria to identify nil-Killing vector fields which generate diffeomorphisms that preserve the set $\mathcal{I}$ for a given metric.

The paper is organized as follows. In section \ref{sec:nilkil}, we determine the general form of a nilpotent self-adjoint operator and relate the choice of frame basis to a preferred null direction, to give a more precise definition for the nil-Killing vector fields. We also show that the nil-Killing vector fields that generalize the Kerr-Schild vector fields form a Lie algebra, and that other Lie algebras can potentially exist depending on the choice of additional conditions for the nil-Killing vector fields. In section \ref{sec:ExistIPD}, we examine the structure of the curvature invariants for a generic spacetime to determine the existence of $IPDs$ and show they can only exist in $\mathcal{I}$-degenerate spacetimes \cite{CHPP2009}. In section \ref{sec:IPDNK}, we employ a frame based approach to determine when a nil-Killing vector field gives rise to an $IPD$ and whether $IPDs$ exist whose infinitesimal generators are not nil-Killing vector fields. In section \ref{sec:IPDKundt}, we consider a general degenerate Kundt spacetime and establish conditions that must be satisfied in order to admit an additional $IPD$. In section \ref{sec:KundtCSI}, we apply the results of section  \ref{sec:IPDKundt} to the Kundt-$CSI$ spacetimes and prove a finite transitive Lie algebra of nil-Killing vector fields which generate $IPDs$ always exists. We summarize our results in section \ref{sec:conclusion} and discuss the existence of $IPDs$ for $\mathcal{I}$-degenerate pseudo-Riemannian manifolds of different signatures. 
 
\end{section}

\begin{section}{Nilpotent Operators and Nil-Killing Vectors} \label{sec:nilkil}

In this section we will introduce some general results about nilpotent operators and relate these results to the alignment classification \cite{classa,classb,classc}, in order to give a more precise definition of a nil-Killing vector field.

\begin{prop} \label{prop:matthew1}
For a spacetime, $(M, {\bf g})$, given ${\bf T}: T_p M \to T_p M$, a self-adjoint endomorphism at an arbitrary point $p \in M$, then
\vspace{3 mm} 

\begin{enumerate}
\item ${\bf T}^2 = 0$ if and only if there exists a null vector, $\bl$, such that ${\bf T}(\{ \bl \}^{\perp} ) = 0$ where $\{ \bl \}^\perp$ denotes the orthogonal vector space to $\bl$. 
\item $\bT^3 =0$ if and only if there exists a null vector $\bl$ such that $\bT( T_p M) \subset \{ \bl \}^\perp$ and $\bT( \{ \bl \}^\perp) \subset \mathbb{R} \bl$. 
\item $\bT$ is nilpotent if and only if $\bT^3 = 0$. 
\end{enumerate}

\end{prop} 
\vspace{1 mm}
\begin{proof} \begin{enumerate}\item Supposing that $\bX \in T_p M$, then $g( \bT\bX, \bT\bX) = g(\bT^2 \bX, \bX) = 0$, and so $\bT\bX$ is a null vector with $\bT\bX \propto \bl$ for some null vector. If ${\bf W} \in \{ \bl \}^\perp$ and $\bZ \in T_p M$, we can write $\bT\bZ = c \bl$ for some constant $c \in \mathbb{R}$, then $${\bf g}(\bT\bW, \bZ) = {\bf g}(\bW,\bT\bZ) = {\bf g}({\bf W}, c\bl) = 0,$$ therefore $\bT{\bf W} = 0$ and $\bT ( \{ \bl \}^\perp) = 0$. To show the other direction, suppose that $\bZ \in T_p M$ and ${\bf W} \in \{ \bl \}^\perp$ then ${\bf g}(\bZ, \bT{\bf W}) = 0$, and so $\bT( T_p M) \subset \mathbb{R} \bl$ which implies $\bT^2 = 0$.

\item We will assume $\bT^2 \neq 0$ and $\bT^3 = 0$ to avoid the first part of the proof. Using the fact that $(\bT^2)^2 = 0$, there must be some null vector $\bl \in T_p M$ such that $\bT^2( \{ \bl \}^\perp )=0$. Given $\bZ \in T_p M$, $${\bf g}(\bT^2 \bZ, {\bf W}) = {\bf g}(\bZ, \bT^2 {\bf W}) = 0,~ \forall~ {\bf W} \in \{ \bl \}^\perp .$$
It follows that $\bT^2 ( T_p M) = \mathbb{R} \bl$, and since $\bT^3 =0$, it is necessary that $\bT \bl = 0$ and $\bT( T_p M) \subset\{ \bl \}^\perp$ since $${\bf g}(\bT\bZ, \bl) = {\bf g}(\bZ, \bT\bl) = 0, \forall~ \bZ \in T_p M.$$ If ${\bf W} \in \{ \bl \}^\perp$ then, $g(\bT^2{\bf W}, {\bf W}) = {\bf g}(\bT{\bf W},\bT{\bf W}) =0$, so that $\bT: \{ \bl \}^\perp \to \{ \bl \}^\perp$ and  $\bT( \{ \bl \}^\perp) = \mathbb{R} \bl$. To prove the other direction, we note that $\bT( T_p M) \subset \{ \bl \}^\perp$ and so ${\bf g}(\bT\bZ, \bl) = {\bf g}(\bZ, \bT\bl) = 0, \forall~ \bZ \in T_p M$ implying $\bT\bl = 0$, and hence $\bT^3 ( T_p M) \subset \bT^2 ( \{ \bl \}^\perp) \subset \mathbb{R} \bT\bl = 0.$

\item Supposing that $\bT^n = 0$ with $\bT^{n-1} \neq 0$ for some $n \geq 3$, then $(\bT^{(n-1})^2 = 0$ and from $(1)$ there is a null-vector $\bl$ such that $\bT^{n-1} ( \{ \bl \}^\perp) =0$. If $\bT(T_p M) \not\subset \{ \bl \}^\perp$,  then there is some non-zero $\bZ \in T_p M$ for which $\bT\bZ \not\in \{ \bl \}^\perp$ giving the identity,
\beq \bT^{(n-1)} (\bT\bZ) = \bT^n \bZ = 0. \nonumber \eeq
\noindent This implies $\bT^{n-1} = 0$ which is a contradiction and so $\bT\b (T_p M) \subset \{ \bl \}^\perp$. Using this fact and ${\bf g}(\bT\bl, \bZ) = {\bf g}(\bl, \bT\bZ) = 0$ implies that $\bT\bl = 0$.

With $\bl$ we can construct a null coframe, $\{\bn, \bm^i,\bl\}$ so that the metric is of the form
\beq {\bf g} = 2 \bl \bn+ \delta_{ij} \bm^i \bm^j\nonumber \eeq
\noindent The self-adjoint operator $\bT$ with $\bT \bl = 0$ will have the matrix representation 
\beq \bT= \left[ \begin{array}{ccc} 0 & {\bf v}^T & c \\ \vdots & {\bf S} & {\bf v} \\ 0 &  \hdots & 0 \end{array} \right] \nonumber \eeq
\noindent where ${\bf v}$ is a $(n-2)$-dimensional vector, ${\bf S}$ is a symmetric $(n-2) \times (n-2)$ matrix and $c$ is real-valued. For any power $k$ of $\bT$, there is some vector ${\bf v}_k$ and real number $c_k$ such that 

\beq \bT^k = \left[ \begin{array}{ccc} 0 & {\bf v}_k^T & c_k \\ \vdots & {\bf S}^k & {\bf v}_k \\ 0 & \hdots & 0 \end{array} \right] \nonumber \eeq
\noindent If $\bT^n = 0$ then ${\bf S}^n = 0$, since ${\bf S}$ is symmetric and is spanned by tensor products of spatial vectors, this implies ${\bf S}=0$ and so for any element of $\{ \bl \}^\perp$, $\bT\{ \bl \}^\perp \subset \mathbb{R} \bl$, giving $\bT^3 = 0$.  We note that if there are two linearly independent null vectors with property $(1)$ or $(2)$ then $\bT=0$. 

\end{enumerate}

\end{proof}

From Proposition \ref{prop:matthew1}, we can give the following definition that motivates the use of the alignment classification. 

\begin{defn}
A self-adjoint endomorphism $\bT$ of the tangent space $T_p M$ of a spacetime is nilpotent with respect to a null vector $\bl$ if $\bT(T_p M) \subset \{ \bl \}^\perp$ and $\bT \bl = 0$. For a particular null vector $\bl$, the collection of self-adjoint nilpotent with respect to $\bl$ will be denoted as $ S_{\bl}( T_p M, {\bf g})$ 
\end{defn}

\noindent This idea can be extended to a symmetric rank two tensor field on a spacetime through the identity 
\beq \bT(\bX,\bY) = {\bf g}(\hat{\bT}\bX, \bY),~~ \forall~ \bX,~ \bY \in  \mathfrak{X}(M), \nonumber \eeq
\noindent such that $\forall ~p \in M$ the endomorphism of $T_pM$, $\hat{\bT}$, is self-adjoint with respect to the metric. If $\bT$ is nilpotent at a point $p \in M$, then there is a corresponding null vector field $\bl \in T_p M$ for which $\bT|_{p}$ can be decomposed in terms of elements of $\{ \bl\}^\perp$. That is, $\bT$ is nilpotent with respect to $\bl$ if $\hat{\bT}$ is nilpotent with respect to $\bl$ which is equivalent to $\bT(\bl, \bZ)=0,~~\forall~ \bZ \in T_p M$ and $\bT({\bf W}, \tilde{{\bf W}})=0,~\forall~ {\bf W}, \tilde{{\bf W}} \in \{ \bl \}^\perp$.  

Due to the smoothness of the manifold, this can be extended in a neighbourhood $U$ of $p \in M$, and so we say a symmetric rank two tensor-field is nilpotent with respect to a null vector field $\bl$, if $\bT|_{p}$ is nilpotent with respect to $\bl |_p$ for all $p \in U$. Completing the null coframe with $\bl$ as a basis element, $\{ \bn, \bl, \bm^i \}$, any nilpotent rank 2 tensor with respect to $\bl$ can be written as

\beq \bT = T_{11} \bl  \bl + 2 T_{1i} \bl \bm^i. \label{nilnull} \eeq

With this definition, we can define a more precise definition of a nil-Killing  vector field \cite{SH2018}: 

\begin{defn}
For a spacetime $(M, {\bf g})$, a vector field $\bX \in \mathfrak{X}(M)$ is nil-Killing with respect to $\bl$ if $\mathcal{L}_{\bX} {\bf g} \in S_{\bl} (T_p M, {\bf g})$.
\end{defn}

\noindent Note that in the Riemannian case, this can only occur if $\bX$ is Killing. In the literature, a specialization of the nil-Killing vector fields known as the Kerr-Schild vector fields, has been discussed  \cite{SC2000, Senov2003}  these are defined as nil-Killing vector fields with respect to $\bl$ for which $\mathcal{L}_{\bX} {\bf g} = {\bf T}$ is  nilpotent of order two, ${\bf T}^2 = 0$ with the additional condition: \beq [\bX, \bl] \propto \bl. \label{KScond} \eeq

\noindent This additional condition allows nil-Killing vector fields to act as automorphisms on $S_{\bl} (T_p M, {\bf g})$.

\begin{lem} \label{lem:matthew1}
Given a non-vanishing null vector field $\bl$ in a spacetime $(M, {\bf g})$ and a nil-Killing vector field, $\bX$, with respect to $\bl$  satisfying $$[\bX, \bl] = f \bl, f \in C^\infty (M).$$ If $\bT \in S_{\bl} (T_p M, {\bf g})$ then $\mathcal{L}_{\bX} \bT \in S_{\bl} (T_p M, {\bf g})$. 
\end{lem}

\begin{proof}  
Suppose that ${\bf W} \in \{ \bl \}^\perp$, then the conditions that $\bX$ is nil-Killing in Proposition \ref{prop:matthew1} implies that $\mathcal{L}_{\bX} {\bf g} (\bl, \{ \bl\}^\perp) =0$ and $\mathcal{L}_{\bX} {\bf g} (\{ \bl\}^\perp, \{ \bl\}^\perp) = 0$ (we have made a minor abuse of notation to treat $\mathcal{L}_{\bX} {\bf g}$ as the corresponding nilpotent operator), the condition \eqref{KScond}  implies that $[\bX, {\bf W}] \in \{\bl \}^\perp$ since 

\beq  0 = \mathcal{L}_{\bX} {\bf g}(\bl, {\bf W}) = - {\bf g}([\bX,\bl],{\bf W}) - {\bf g}(\bl, [\bX, {\bf W}]) = - {\bf g}(\bl, [\bX, {\bf W}]) . \nonumber \eeq
\noindent For any $\bZ \in \mathfrak{X}(M)$ and  ${\bf W}, \tilde{{\bf W}} \in \{\bl \}^\perp$ this implies

\beq \mathcal{L}_{\bX} \bT( \bl, \bZ) = \bX(\bT(\bl, \bZ)) - \bT([\bX,\bl], \bZ) - \bT(\bl, [\bX, \bZ]) = 0, \eeq
\noindent  and 
\beq \mathcal{L}_{\bX} \bT( {\bf W}, \tilde{{\bf W}}) = \bX(\bT({\bf W}, \tilde{{\bf W}})) - \bT([\bX, {\bf W}], \tilde{{\bf W}}) - \bT({\bf W}, [\bX,\tilde{{\bf W}}]) = 0. \eeq
\noindent Therefore $\mathcal{L}_{\bX} \bT$ is also nilpotent with respect to $\bl$. 
\end{proof}

\noindent We can now show that the set of nil-Killing vector fields satisfying \eqref{KScond} form a Lie algebra and not just the Kerr-Schild vector fields.

\begin{prop} \label{prop:matthew2}
For any spacetime $(M, {\bf g})$ and $\bl$ a null vector field then
\beq \mathfrak{g}_{\bl} := \{ \bX \in \mathfrak{X}(M) | [\bX, \bl] \propto \bl, ~ \bX \text{ is nil-Killing with respect to } \bl \} \nonumber \eeq 
\noindent is a Lie algebra and 
\beq \mathfrak{h}_{\bl} := \{ \bX \in \{ \bl \}^\perp  | [\bX, \bl] \propto \bl, ~ \bX \text{ is nil-Killing with respect to } \bl \} \nonumber \eeq
\noindent is an ideal. 
\end{prop} 

\begin{proof}
Suppose that $\bX, \bY \in \mathfrak{g}_{\bl}$, then by assumption $\mathcal{L}_{\bY} {\bf g}$ is nilpotent with respect to $\bl$, and so by Proposition \ref{prop:matthew1} $\mathcal{L}_{\bX} (\mathcal{L}_{\bY} {\bf g})$ is nilpotent since $[ \bX, \bl ] \propto \bl$. Repeating this argument with $\bX$ and $\bY$ switched gives another nilpotent operator, and the difference 
\beq \mathcal{L}_{[\bX,\bY]} {\bf g} = \mathcal{L}_{\bX} (\mathcal{L}_{\bY} {\bf g}) - \mathcal{L}_{\bY}( \mathcal{L}_{\bX} {\bf g}) \nonumber \eeq
\noindent must be nilpotent with respect to $\bl$ as well. While $[\bX,\bY]$ is nil-Killing with respect to $\bl$ the condition \eqref{KScond} must be preserved also. Supposing that $f_1, f_2 \in C^\infty (M)$ such that $[\bX,\bl] = f_1 \bl$ and $[\bY,\bl] = f_2 \bl$ the Jacobi identity gives
\beq [[\bX,\bY], \bl] &=& - [[\bY,\bl], \bX] - [[\bl, \bX], \bY] = [\bX, f_2 \bl] - [\bY, f_1 \bl] \nonumber \\ &=& (\bX(f_2) - \bY(f_1)) \bl,  \nonumber \eeq
\noindent therefore $[\bX,\bY] \in \mathfrak{g}_{\bl}$ and $\mathfrak{g}_{\bl}$ is a Lie algebra

Suppose now that $\bX \in \mathfrak{g}_{\bl}$ and $\bY \in \mathfrak{h}_{\bl}$, then $[\bX,\bY] \in \{\bl\}^\perp$ and so $[\bX,\bY] \in \mathfrak{h}_{\bl}$ implying that $\mathfrak{h}_{\bl}$ is an ideal. 
\end{proof}

\noindent If $\{\bl\}^\perp$ is integrable, the condition in equation \eqref{KScond} can be relaxed for $\mathfrak{h}_{\bl}$. This will be particularly important for the degenerate Kundt spacetimes which admit an integrable $\{\bl\}^\perp$ and cannot be uniquely characterized locally by their SPIs. For such spacetimes, we expect that a subset of the nil-Killing vector fields to give rise to $IPDs$ and that they should form a Lie algebra.

\begin{cor} \label{cor:matthew1}

For any spacetime $(M, {\bf g})$ and $\bl$ a null vector field, if $\{\bl\}^\perp$ is integrable then
\beq \mathfrak{h}_{\bl} := \{ \bX \in \{ \bl \}^\perp  | \bX \text{ is nil-Killing with respect to } \bl \}  \eeq

\noindent is a Lie algebra. 
\end{cor}

\begin{proof}
If $\{{\bl} \}^\perp$ is integrable, and $$\bZ \in \{ \bX \in \{ {\bl} \}^\perp  | \bX \text{ is nil-Killing with respect to } {\bl} \},$$ then $[\bZ,{\bf W}] \in \{{\bl}\}^\perp,~\forall~ {\bf W} \in \{{\bl}\}^\perp$. Since $\bZ$ is nil-Killing with respect to ${\bl}$ it follows that 
\beq 0= \mathcal{L}_{\bZ} {\bf g} ({\bl}, {\bf W}) = -{\bf g}([\bZ,\bl], {\bf W}) - {\bf g}(\bl, [\bX,{\bf W}]) = - {\bf g}([\bZ,\bl], {\bf W}),  \nonumber \eeq
\noindent and so $[\bZ,\bl] \propto \bl$ and $\bZ \in \mathfrak{h}_{\bl}$. The converse inclusion is trivial, and hence $\mathfrak{h}_{\bl}$ is a Lie algebra. 
\end{proof}

Alternatively, we can relax the condition in equation \eqref{KScond} and instead consider any nil-Killing vector field, $\bX$, for which $\bT =\mathcal{L}_{\bX} {\bf g}$ is nilpotent with respect to $\bl$ and $\bT^2 = 0$. We will say {\it $\bX$ is a nil-Killing vector field with respect to $\bl$ of order two}.

\begin{prop} \label{prop:matthew3}
Given a  spacetime $(M, {\bf g})$ and a non-vanishing null vector field $\bl$. If $\bX, \bY \in \mathfrak{X}(M)$ are nil-Killing with respect to $\bl$ of order two, then $[\bX,\bY]$ is nil-Killing with respect to $\bl$.
\end{prop}

\begin{proof}
From Proposition \ref{prop:matthew1}, a vector field $\bZ$ is nil-Killing with respect to $\bl$ of order two if $$ \mathcal{L}_{\bZ} {\bf g} ({\bf W},{\bf P}),~~\forall~ {\bf W} \in \{ \bl\}^\perp,~{\bf P} \in \mathfrak{X}(M). $$
\noindent In addition for any nil-Killing vector field, $\bZ$, with respect to $\bl$, $[\bZ, \bl] \in \{ \bl \}^\perp$ since
\beq \mathcal{L}_{\bZ} {\bf g}(\bl, \bl) = -2 {\bf g}([\bZ,\bl], \bl) = 0. \nonumber \eeq
\noindent Using these facts it follows that 
\beq & \mathcal{L}_{[\bX,\bY]} {\bf g}(\bl, {\bf P}) = 0,~~\forall~ {\bf P} \in \mathfrak{X}(M), & \nonumber \\
&   \mathcal{L}_{[\bX,\bY]} {\bf g}({\bf W}, \tilde{{\bf W}}) = 0,~~\forall~ {\bf W},~\tilde{{\bf W}} \in \{ \bl \}^\perp & \nonumber \eeq
\noindent From Proposition \ref{prop:matthew1}, this implies that $[\bX,\bY]$ is nil-Killing with respect to $\bl$. 
\end{proof}

\noindent While $[\bX,\bY]$ is nil-Killing with respect to $\bl$ it may no longer be a nil-Killing vector field of order two, and so nil-Killing vector fields of this type do not form a Lie algebra without imposing additional conditions on the metric or the set of nil-Killing vector fields $\bX$ and $\bY$. For example, in the case of Kerr-Schild vector fields requiring $[\bX,\bl] \propto \bl$ and $[\bY,\bl] \propto \bl$ forces $[\bX,\bY]$ to be a nil-Killing vector field of order two and hence forms a Lie algebra. 

This suggest that there are other Lie algebras within the set of nil-Killing vector fields. It is of interest to determine if a condition can be imposed to produce a finite Lie algebra for the nil-Killing vector fields. We are primarily interested in determining a finite Lie algebra of nil-Killing vector fields that generate a transitive set of $IPDs$, as such we will employ our characterization of nil-Killing vector fields to determine when a nil-Killing vector field preserves the set $\mathcal{I}$. 

\end{section}

\begin{section}{Existences of $\mathcal{I}$-Preserving Diffeomorphisms} \label{sec:ExistIPD}

Due to our interest in $IPDs$, we would like to find all vector fields ${\bf X}$ such that $$ \mathcal{L}_{\bX} I = 0,~\forall~I \in \mathcal{I} $$
\noindent and which are {\bf not} Killing vector fields, we will call ${\bf X}$ an {\it $IPD$ infinitesimal generator}, or an {\it $IPD$ vector field}. In order to do so, we will employ an alternative set of invariants that locally characterize a spacetime uniquely: the {\it Cartan invariants}, $\mathcal{R}^q$, which are the components of the curvature tensor and its covariant derivatives relative to a particular frame determined by the Cartan-Karlhede algorithm. A review of the Cartan-Karlhede algorithm is outside the scope of the current paper, we will refer to Chapter 9 of \cite{kramer} for the 4D implementation of the algorithm and \cite{CKAHD, MMCetal2018} for a discussion of the algorithm in five and higher dimensions. An $IPD$ vector field exists when the set of Cartan invariants $\mathcal{R}^q$ has a larger rank (i.e., the number of functionally independent components) than $\mathcal{I}$, $$ rank ( \mathcal{R}^q) > rank ( \mathcal{I}).$$ 

%xxyyzz 
This condition implies that the spacetime is not locally characterized uniquely by its $SPIs$. In 3D and 4D, such metrics must belong to the degenerate Kundt class and the curvature tensor and its covariant derivatives must be of type {\bf II} to all orders \cite{CHPP2009}. In higher dimensions it is conjectured that this will be the case as well. Denoting $[\mathcal{R}]_{b.w. 0} $ as the set of components of the curvature tensor and its covariant derivatives of b.w. zero, we can introduce an alternative criteria for the existence of $IPD$ vector fields for all $\mathcal{I}$-degenerate spacetimes using the alignment classification  without generating the entire set $\mathcal{I}$. 

\begin{thm} \label{thm:LorExist}
Relative to the basis determined by the Cartan-Karlhede algorithm, a spacetime admits a non-trivial $IPD$ vector field, ${\bf X}$, such that 
\beq \mathcal{L}_{\bX} \mathcal{I} = 0, \eeq 
if and only if 
the spacetime is of alignment type {\bf II} to all orders and 
$$ 0 \leq rank ( [\mathcal{R}^q]_{b.w. 0} ) < rank ( \mathcal{R}^q).$$
\noindent That is, the spacetime is $\mathcal{I}$-degenerate.
\end{thm}

\begin{proof}
If a non-trivial $IPD$, ${\bf X}$, exists then we may choose local coordinates where ${\bf X} = \frac{\partial}{\partial x}$ implying that the $SPIs$ are independent of $x$. We note that since we have assumed that ${\bf X}$ is a non-trivial $IPD$ vector field, there must exist some Cartan invariant which is dependent on the $x$ coordinate \cite{kramer}. As we cannot express the $x$ coordinate in term of $SPIs$, the $SPIs$ are unable to distinguish orbits of ${\bf X}$, and so they do not uniquely characterize the spacetime. Thus, the spacetime is necessarily $\mathcal{I}$-degenerate, and Corollary 3.4  in \cite{Hervik2011} implies that the curvature tensor and its covariant derivatives cannot be of alignment type {\bf I} or {\bf G} at any order. That is, the spacetime is at least of alignment type {\bf II} to all orders. For any such spacetime, the components of b.w. zero of the curvature tensor are determined by the $SPIs$ (Corollary II.11 in \cite{CHP2010}) and hence 
$$ 0 \leq rank ( [\mathcal{R}^q]_{b.w. 0} ) < rank ( \mathcal{R}^q).$$
\noindent 
The opposite direction follows from the fact that when constructing a complete contraction of any tensor of type {\bf II}, only the b.w. 0 components contribute to the resulting $SPI$.

\end{proof}

The dimension of the Lie group of isometries, $G$, can be computed from the Cartan-Karlhede algorithm using the formula:  $$ dim(G) = D - I_q + dim (H_q),$$
\noindent where $D$ is the dimension of the manifold, $q$ is the final iteration of the algorithm, $I_q$ is the number of functionally independent Cartan invariants and $H_q$ is the linear isotropy group. Motivated by this result we can determine the dimension, $m$, of the subset of the tangent space spanned by all non-trivial $IPD$ vector fields by taking the difference:
$$ m = I_q - rank ( [\mathcal{R}^q]_{b.w. 0} ).$$
\noindent For example, for a generic degenerate Kundt spacetime admitting no additional isometries, $m =1$, whereas for a $CSI$ spacetime that is not locally homogeneous $m = D$.

\end{section}

\begin{section}{The Nil-Killing Condition and $IPD$ Vector Fields} \label{sec:IPDNK} 

Assuming the spacetime is $\mathcal{I}$-degenerate, let us consider the condition introduced in \cite{SH2018} to study the set of $IPD$ vector fields, $$\mathcal{L}_{\bX} {\bf g} = {\bf N},$$
\noindent where ${\bf N}$ is a nilpotent rank two tensor. For any nilpotent operator, there is a related null direction, $\bl$, as illustrated in equation \eqref{nilnull}.  Let us choose $\bl$ as a coframe basis element, and complete the coframe basis $ \{ \mbold{\theta}^a\} = \{ \bn,\bl, \bm^i\}$ then we can consider the effect of a Lie derivative in the direction of ${\bf X}$ on the coframe basis:
\beq \begin{aligned} & \mathcal{L}_{\bX} \bl = A \bl + \tilde{A} \bn + B_i \bm^i, \\ & \mathcal{L}_{\bX} \bn = C \bl + \tilde{C} \bn + D_i \bm^i, \\ & \mathcal{L}_{\bX} \bm^i = E^i \bl +\tilde{E}^i \bn + F^i_{~j} \bm^j, \end{aligned} \label{action} \eeq

\noindent where the coefficients are functions of the coordinates. Imposing the condition that $\mathcal{L}_{\bX} {\bf g}$ is nilpotent implies that this symmetric tensor must only have non-zero components with negative b.w. which puts conditions on the coefficients

\beq \tilde{A} = 0,~\tilde{C} = -A,~ \tilde{E}^i = -B^i \text{ and } F_{ij} = -F_{ji}, \eeq

\noindent and so
\beq \begin{aligned} & \mathcal{L}_{\bX} \bl = A \bl + B_i \bm^i, \\ & \mathcal{L}_{\bX} \bn = C \bl - A  \bn + D_i \bm^i, \\ & \mathcal{L}_{\bX} \bm^i = E^i \bl - B^i \bn + F^i_{~j} \bm^j,~~F_{(ij)} = 0. \label{preipd} \end{aligned} \eeq

We will now focus our attention on nil-Killing vector fields such that $\mathcal{L}_{\bX} {\bf g}$ is nilpotent with respect to the vector field $\bl$ for which the Riemann tensor and its covariant derivatives are of type {\bf II} or higher \cite{CHPP2009}. Using abstract index notation briefly, we will consider the subset of these nil-Killing vector fields which also annihilate $SPIs$ constructed from an arbitrary rank two symmetric curvature tensor $R_{ab}$ (such as the Ricci tensor or $C_{abcd;e} C^{abcd;f}$ as two examples) with the simplest $SPI$, the contraction $$I = R^a_{~a}.$$ 

In general, for a degenerate Kundt spacetime, the Riemann tensor and its covariant derivatives are of alignment type {\bf II}, and so the Ricci and Weyl tensor are at least of alignment type {\bf II}. We will assume it is possible to construct at least one rank two tensor of alignment type {\bf II}. We note that this analysis will be restricted to degenerate Kundt spacetimes of  alignment type {\bf II}, {\bf III} and {\bf N} to all orders, and that the subclass of metrics which have alignment type {\bf D} to all orders, known as type ${\bf D}^k$ will be excluded. Such metrics are $\mathcal{I}$-degenerate but are characterized by their $SPIs$, although not uniquely, since $rank([\mathcal{R}]_{b.w. 0} ) = rank( \mathcal{R})$.
%xxyyzz

If $ \mathcal{L}_{\bf X} I = 0$, then the trace of $\mathcal{L}_{\bf X} R_{ab}$ is zero since the Lie derivative commutes with contraction.  In order to avoid the  possibility that $\mathcal{L}_{\bf X} R_{ab} $ could be trace-free for some choices of $R_{ab}$, we assume the following property holds: 
\begin{defn} 
A spacetime is {\it generic of type} {\bf II}, {\bf D}, {\bf III} or {\bf N} if the set of rank two curvature tensors spans the vector space of rank two tensors of alignment type {\bf II}, {\bf D}, {\bf III} or {\bf N} respectively. 
\end{defn}

\noindent For any $\mathcal{I}$-degenerate spacetime  which is generic of type {\bf II}, corollary II.11 in \cite{CHP2010} gives a necessary condition for the vanishing of the trace of $\mathcal{L}_{\bf X} R_{ab}$:
\beq [ \mathcal{L}_{\bf X} R_{ab} ]_{b.w. 0} = 0, \nonumber \eeq
\noindent or in standard notation,  $$[ \mathcal{L}_{\bf X} (R_{ab} \mbold{\theta}^a \mbold{\theta}^b)]_{b.w. 0} = 0.$$ That is,  $\mathcal{L}_{\bf X} R_{ab}$ is of type {\bf III} as all b.w. zero components must vanish.

%xxyyzz Recheck theta and abstract index notation.

Imposing the condition that $$\mathcal{L}_{\bf X} (R_{ab} \mbold{\theta}^a \mbold{\theta}^b)  = {\bf X}(R_{ab}) \mbold{\theta}^a \mbold{\theta}^b  + 2 R_{ab} \mbold{\theta}^a \mathcal{L}_{\bf X} \mbold{\theta}^b $$ is of type {\bf III}, we note that ${\bf X}(R_{ab}) \mbold{\theta}^a \mbold{\theta}^b$ will only contribute negative b.w. terms to this tensor sum. Therefore, we have additional conditions on the Lie derivative of the basis \eqref{preipd} by ${\bf X}$:
\beq \begin{aligned} & \mathcal{L}_{\bf X} \bl = A \bl, \\ & \mathcal{L}_{\bf X} \bn = C \bl - A  \bn + D_i \bm^i, \\ & \mathcal{L}_{\bf X} \bm^i = E^i \bl + F^i_{~j} \bm^j,~~F_{(ij)} = 0 \end{aligned}  \label{preipd0} \eeq

\noindent where $F_{ij}$ satisfies the supplemental condition that $R_{(i|j} F^j_{~k)} = 0$. This will hold for all symmetric rank two tensors that can be constructed from the curvature tensor and its covariant derivatives. Furthermore, since $\mathcal{L}_{\bf X} R_{ab}$ is of type {\bf III} and $R_{ab}$ is at least of type {\bf II}, any $SPI$ constructed from contractions of copies of $R_{ab}$ will vanish under $\mathcal{L}_{\bf X}$ due to the Liebnitz property and the fact that the Lie derivative of the tensor product, $\mathcal{L}_{\bf X} (R_{a b}  )$, must be of type {\bf III}. 

Repeating this analysis to tensors constructed from the curvature tensor and its covariant derivatives of higher rank yield no additional constraints. However, applying the analysis for those $\mathcal{I}$-degenerate spacetimes whose Riemann tensors and their covariant derivatives are of alignment type {\bf III} and {\bf N} gives the following result.

\begin{prop}
For any $\mathcal{I}$-degenerate spacetime which is at least  generic of type {\bf II}, suppose that ${\bf X}$ is a nil-Killing vector field with respect to the Riemann-aligned null vector field, $\bl$. If under exponentiation of {\bf X} all $SPIs$ constructed from the curvature tensor and its covariant derivatives are preserved, then $\mathcal{L}_{\bf X}$ produces the following transformation on the coframe basis: 

\begin{itemize} 
\item Alignment type {\bf II} and {\bf III}:
\beq \begin{aligned} & \mathcal{L}_{\bf X} \bl = A \bl, \\ & \mathcal{L}_{\bf X} \bn = C \bl - A  \bn + D_i \bm^i, \\ & \mathcal{L}_{\bf X} \bm^i = E^i \bl + F^i_{~j} \bm^j,~~F_{(ij)} = 0. \end{aligned} \label{LieDII} \eeq
\item Alignment type {\bf N}:
\beq \begin{aligned} & \mathcal{L}_{\bf X} \bl = A \bl + B_i \bm^i, \\ & \mathcal{L}_{\bf X} \bn = C \bl - A  \bn + D_i \bm^i, \\ & \mathcal{L}_{\bf X} \bm^i = E^i \bl - B^i \bn + F^i_{~j} \bm^j,~~F_{(ij)} = 0.  \end{aligned} \label{LieDN}  \eeq

\end{itemize}
\end{prop}

In fact, using the action \eqref{action} of the Lie derivative of the coframe in the direction of a vector field ${\bf X}$ we may prove the following result:

\begin{prop} \label{prop:IPDtypeII}
For any $\mathcal{I}$-degenerate spacetime which is generic of type {\bf II}, an $IPD$ vector field is necessarily a nil-Killing vector field with respect to $\bl$ of the form \eqref{preipd0}. 
\end{prop}

\noindent If the $\mathcal{I}$-degenerate spacetime admits curvature tensors of type {\bf III} or {\bf N}, the set of $IPD$ vector fields may not necessarily be contained within the set of nil-Killing vector fields since the action of the Lie derivative in the direction of a vector field ${\bf X}$ on the coframe basis will not give enough b.w. zero components to restrict the form of \eqref{action}.

\end{section}
  
\begin{section}{$\mathcal{I}$-Preserving Diffeomorphisms in the Kundt Spacetimes} \label{sec:IPDKundt}

As the Kundt spacetimes contain a subclass that are $\mathcal{I}$-degenerate, we will study the curvature structure of this subclass to determine conditions on the metric functions in order to admit an additional $IPD$ vector field, ${\bf X}$. The class of Kundt spacetimes are given by the line element:
\beq d s^2=  2 du \left(d v+H(v,u,x^\delta) d u+ W_{\alpha}(v,u,x^\delta )d x^\alpha \right)+ \tilde{g}_{\alpha \beta}(u,x^\delta) dx^\alpha d x^\beta. \label{Kundt} \eeq

\noindent Choosing the initial null coframe, 
\beq \bl = du,~ \bn = dv + Hdu + W_i m^i_{~\alpha} dx^\alpha,~~ \bm^i = m^i_{~\alpha} dx^\alpha, \label{Kundtframe} \eeq
\noindent such that the metric tensor in the line element, $ds^2 = g_{\alpha \beta} dx^\alpha dx^\beta$, takes the form
\beq g_{\alpha \beta} = 2 \ell_{(\alpha} n_{\beta)} + \delta_{ij} m ^i_\alpha m^i_\beta. \nonumber \eeq
\noindent We apply a Lorentz transformation to work with the coframe arising from the Cartan-Karlhede algorithm. While there may be some isotropy remaining from this choice, this will have no effect on the resulting analysis of the b.w. zero components as $\bl$ remains fixed. This will allow us to consider the b.w. zero components of the initial frame with $\bm^i$ adapted to the geometry of the transverse metric.

We note that $\bl = \frac{\partial}{\partial v}$ is a nil-Killing vector field, 
$$ \mathcal{L}_{\bl} {\bf g} =  H_{,v} \bl \bl + 2 W_{i,v} \bm^i \bl,$$
\noindent but it is not necessarily an $IPD$ vector field as the  linearly independent non-zero components of the Riemann tensor with b.w. 1 and 0 are: 

\beq R_{ 1 2 1 i}&=&-\frac 12 W_{ i,vv} \nonumber \\
R_{ 1 2 1 2}&=& -H_{,vv}+\frac 14\left(W_{{i},v}\right)\left(W^{ i,v}\right), \nonumber \\
R_{ 1 2 i j}&=& W_{[ i}W_{ j],vv}+W_{[ i; j],v}, \label{RiemBW0} \\
R_{ 1 i 2 j}&=& \frac 12\left[-W_{ j}W_{ i,vv}+W_{ i; j,v}-\frac 12 \left(W_{  i,v}\right) \left(W_{ j,v}\right)\right], \nonumber \\
R_{ i j kl}&=&\tilde{R}_{ i j k l}. \nonumber \eeq 

\noindent where $\tilde{R}_{ijkl}$ denotes the curvature tensor of the transverse space.

In order for the metric to be degenerate Kundt, it must be of type {\bf II} to all orders.  From \cite{CHPP2009, CHP2010} this occurs if and only if both of the following quantities vanish:

\beq \begin{aligned} I_0 = R^{abcd} R^{~e~f}_{a~c} \mathcal{L}_{\bl} \mathcal{L}_{\bl} g_{bd} \mathcal{L}_{\bl} \mathcal{L}_{\bl} g_{ef},~~ K_{ab} = \mathcal{L}_{\bl} \mathcal{L}_{\bl} \mathcal{L}_{\bl} g_{ab}, \end{aligned} \eeq

\noindent which gives the following conditions on the metric functions $H$ and $W_i$ in \eqref{Kundtframe}:

\beq \begin{aligned} H &= H^{(2)}(u,x^\delta) \frac{v^2}{2} + H^{(1)} (u,x^\delta) v + H^{(0)}(u,x^\delta),  \\  W_i &= W_i^{(1)}(u,x^\delta) v + W_i^{(0)}(u,x^\delta). \end{aligned} \label{degKundt} \eeq

\noindent Looking at the b.w. zero components of the Riemann tensor \eqref{RiemBW0}, it is clear that all of the components are now independent of $v$. However, since $$\mathcal{L}_{\bl} {\bf g} \neq 0$$ this implies that there are components of the Riemann tensor that are dependent on the $v$ coordinate, namely the negative b.w. terms. 

In section \ref{sec:IPDNK} we have shown that the condition that ${\bf X}$ is a nil-Killing vector field is not sufficient to prove it is an $IPD$ vector field. If a degenerate Kundt spacetime admits an $IPD$ vector field, we can deduce its properties from its action on the b.w. zero components treated as Cartan invariants. Supposing there is an additional $IPD$ vector field ${\bf X}$, we can determine conditions from the equations \eqref{RiemBW0} that the functions $H^{(2)}$ and $W^{(1)}_i$ must satisfy at zeroth order: 

\beq \begin{aligned}  & \sigma(u,x^\delta) = H_{,vv} - \frac{1}{4} W_{i,v} W^i_{~,v}, \\ 
& a_{ij}(u,x^\delta) = W_{[i;j],v}, \\
& s_{ij}(u,x^\delta) = W_{(i;j),v} - \frac12 W_{i,v} W_{j,v}, \\
& \mathcal{L}_{\bf X} R_{ijkl} = \mathcal{L}_{\bf X} \sigma = \mathcal{L}_{\bf X} a_{ij} = \mathcal{L}_{\bf X} s_{ij} =0. \end{aligned} \label{IPDeqn0} \eeq

To determine additional conditions we compute the first covariant derivative of the curvature tensor. We note that the connection coefficients cannot contribute positive b.w terms, and the connection coefficients with zero b.w. are independent of $v$, this implies that the b.w. zero components of the covariant derivatives of the Riemann tensor will also be independent of $v$ \cite{CHPP2009}. Thus, $\bl = \frac{\partial}{\partial v}$ will be an $IPD$. 

Due to the form of $R_{ijkl}$ and the connection coefficients for the degenerate Kundt metrics, we can identify a simple condition any $IPD$ vector field must satisfy in terms of the  covariant derivatives of the transverse curvature tensor:

\beq \mathcal{L}_{\bf X} R_{ijkl} = \mathcal{L}_{\bf X} R_{ijkl;i_1} = \mathcal{L}_{\bf X} R_{ijkl; i_1 \cdots i_p } = 0. \nonumber \eeq

\noindent This implies that ${\bf X}$ must be a Killing vector field for the transverse space and so $\mathcal{L}_{\bf X} \Gamma^i_{jk} = 0$ as well. The remaining first order $IPD$ equations are then

\beq \begin{aligned} &\mathcal{L}_{\bf X} \mathcal{L}_{\bm_i} \sigma = \mathcal{L}_{\bf X} \mathcal{L}_{\bm_k} a_{ij} = \mathcal{L}_{\bf X} \mathcal{L}_{\bm_k} s_{ij} = 0, \\
& \alpha_i = R_{121i;2} =  \sigma W_{i,v} - \frac12 (s_{ij} + a_{ij}) W^{j,v}, \\
& \beta_{ijk} = R_{1ijk;2} = W^{l,v} \tilde{R}_{lijk} - W_{i,v} a_{jk} + (s_{i[j}+a_{i[j})W_{k],v}, \\
& \mathcal{L}_{\bf X} \alpha_i = \mathcal{L}_{\bf X} \beta_{ijk} = 0. \end{aligned} \label{IPDeqn1} \eeq

 Higher order covariant derivatives provide additional conditions on the vielbein of the transverse space and the b.w. 0 components of the curvature tensor and its covariant derivatives. If ${\bf X}$ is an $IPD$ vector field and $H \in [\mathcal{R}^q]_{b.w. 0}$,  then 
\beq \mathcal{L}_{[{\bf X}, {\bf m}_i]} H = 0. \nonumber \eeq 
\noindent This provides a consistency condition for the zeroth and first order equations. For any degenerate Kundt spacetime, the action of the Lie derivative of ${\bf X}$ acting on the vielbein takes the form in equation \eqref{LieDII} or \eqref{LieDN} and so the consistency condition implies that the anti-symmetric matrix $F_{ik}$ accounts for elements of the isotropy group of the transverse metric $\tilde{g}_{ij}$. 

In summary, we have the following theorem:
 
\begin{thm} \label{thm:IPDKundt}

For a degenerate Kundt spacetime, if ${\bf X}$ is an $IPD$ vector field, then the metric functions $H^{(2)}$ and $W^{(1)}_i$ must satisfy the first order and second order equations \eqref{IPDeqn0} and \eqref{IPDeqn1} while the transverse space $\tilde{{\bf g}}$ admits ${\bf X}$ as a Killing vector field.  
\end{thm}

\noindent This result is in agreement with theorem II.7 in \cite{CHP2010} which states that all b.w. zero components of the curvature tensor and its covariant derivatives depend on $H^{(2)}$ and $W^{(1)}_i$ alone. That is, we can ignore the lower order $v$ coefficients $H^{(0)}$ and $W^{(0)}_i$ in the metric. Theorem \ref{thm:IPDKundt} leads to the following corollary:

\begin{cor} \label{cor:IPDKundt}
If the degenerate Kundt metric ${\bf g'}$ with $H^{(1)} = H^{(0)} = W^{(0)}_i = 0$ admits an $IPD$ vector field, ${\bf X}$, then any related metric ${\bf g}$ with non-zero $H^{(1)}, H^{(0)}$ or $W^{(0)}_i$ will also admit ${\bf X}$ as an $IPD$ vector field.
\end{cor} 

\noindent For a degenerate Kundt spacetime, if a particular vector field ${\bf X}$ is chosen as an $IPD$ vector field, for example ${\bf \tilde{X}} = \frac{\partial}{\partial u}$, all solutions to the equations \eqref{IPDeqn0} and \eqref{IPDeqn1} for $W^{(1)}_i$ may be difficult to determine. However, a simple solution can always be produced by requiring that 

\beq \tilde{g}_{\alpha \beta}(x^\delta) \text{ and } \mathcal{L}_{{\bf \tilde{X}}} W^{(1)}_i = 0. \nonumber \eeq

\noindent In this case, ${\bf \tilde{X}}$ will be a Killing vector field for the degenerate Kundt metric, ${\bf g'}$. From this observation we have another corollary:

\begin{cor} \label{cor:KVKundt}
If the degenerate Kundt metric ${\bf g'}$ with $H^{(1)} = H^{(0)} = W^{(0)}_i = 0$ admits a Killing vector field ${\bf X}$, i.e., 
\beq \mathcal{L}_{\bf X}{\bf g}' = 0, \label{IPDcrit} \eeq
\noindent then any related metric ${\bf g}$ with non-zero $H^{(1)}, H^{(0)}$ or $W^{(0)}_i$ will admit a nil-Killing $IPD$ vector field. 
\end{cor} 
\noindent We note that for some choices of $H^{(1)}, H^{(0)}$ and $W^{(0)}_i$ the vector field ${\bf X}$ will still be a Killing vector field for ${\bf g}$.
\end{section} 

\begin{section}{Kundt-$CSI$ Spacetimes} \label{sec:KundtCSI}

A Kundt-$CSI$ spacetime is a degenerate Kundt spacetime where the transverse space $\tilde{{\bf g}}$ is a locally homogeneous Riemannian manifold and the metric functions $H^{(2)}$ and $W^{(1)}$ satisfy the equations \eqref{IPDeqn0} and \eqref{IPDeqn1} with $\sigma, a_{ij}, s_{ij}, \alpha_i$ and $\beta_{ijk}$ constant. In 3D and 4D any $CSI$ spacetime is either locally homogeneous or belongs to the Kundt-$CSI$ class \cite{CSI4c, CSI4b} while in higher dimensions the non-flat $VSI$ spacetimes are a subset of the Kundt-$CSI$ spacetimes \cite{Higher} and it is conjectured that all higher dimensional $CSI$ spacetimes are either locally homogeneous or belong to the Kundt-$CSI$ class as well.

Using Proposition \ref{prop:IPDtypeII} and Corollary \ref{cor:KVKundt} we are able to confirm the conjecture for Kundt-$CSI$ spacetimes in \cite{SH2018}:

\begin{con}
Assume that a $D$-dimensional spacetime has all constant
curvature invariants ($CSI$). Then there exists a set $N$ of nil-Killing vector fields which is transitive; i.e., $dim(N|p) = D$ for all $p \in M$.
\end{con}

Given a  Kundt-$CSI$ spacetime, we may use a  diffeomorphism $\phi_t$ with respect to a point $p$ generated by the boost defined in \cite{SHY2015} and take the limit $\lim_{t \to \infty} \phi_t^* {\bf g} = {\bf g'}$ to produce a locally homogeneous Kundt-$CSI$ spacetime known as a {\it Kundt$^\infty$ triple} with non-zero metric functions \eqref{degKundt} of the form: 

\beq (H, {\bf W}, {\bf \tilde{g}}) = (H^{(2)}(u_0,x^\delta), W^{(1)}_i(u_0,x^\delta), \tilde{g}_{\gamma \epsilon}(u_0, x^\delta) dx^\gamma dx^\epsilon). \eeq

\noindent Thus, a corresponding locally homogeneous  Kundt$^\infty$  triple which is generic of type {\bf D} can always be generated with an identical set of constant $SPIs$ as the original Kundt-$CSI$ spacetime.

In fact, the locally homogeneous Kundt$^\infty$  triples are of alignment type ${\bf D}^k$, i.e., the curvature tensor and its covariant derivatives are of type ${\bf D}$ to all orders \cite{SM2018}. As the $SPIs$ fully determine the Cartan invariants of a type ${\bf D}^k$ spacetime \cite{CHPP2009, SM2018,CHP2010}, the Cartan invariants must be constant, ensuring the existence of a fully transitive set of Killing vector fields. Thus, for any Kundt-$CSI$ spacetime, corollary \ref{cor:KVKundt} and proposition \ref{prop:IPDtypeII} give the following proposition:

\begin{prop}
For any Kundt-$CSI$ spacetime, the Killing vector fields of the corresponding Kundt$^\infty$ triple act as a finite transitive set of nil-Killing $IPD$ vector fields for the original spacetime.
\end{prop}

\end{section}

\begin{section}{Discussion and Future Work} \label{sec:conclusion}

In this paper we have examined the general form of the nilpotent operators and introduced a new definition for the nil-Killing vector fields. Using this definition we have shown that the nil-Killing vector fields, which generalize the Kerr-Schild vector fields, form a Lie algebra. We have also argued that other Lie algebras can be formed by imposing additional conditions on the nil-Killing vector fields. Since the existence of a nil-Killing vector field does not ensure that it will be an $IPD$ vector field, we then studied the existence of $IPD$ vector fields using a frame based approach. By considering the form of the curvature tensor and its covariant derivatives arising from the Cartan-Karlhede algorithm we have determined the dimension of the subset of the tangent space spanned by the $IPD$ vector fields.

Employing the stronger definition of a nil-Killing vector field and the action of the Lie derivative of the nil-Killing vector fields on the coframe, we have shown that the set of nil-Killing vector fields contain $IPD$ vector fields. Furthermore we proved that for a spacetime which is generic of type {\bf II} to all orders, the $IPD$ vector fields are strictly contained in the set of nil-Killing vector fields. In the case of spacetimes which are generic of type {\bf III} or {\bf N} we are unable to show that an $IPD$ vector field is necessarily a nil-Killing vector field, and so it is possible that such algebraically special  $\mathcal{I}$-degenerate spacetimes can admit $IPD$ vector fields which are {\it not} nil-Killing.

To determine the existence of an $IPD$ vector field in a general degenerate Kundt spacetime we have proposed a constructive approach by assuming that an $IPD$ vector field is given and determining the form of the metric functions. The existence of an $IPD$ vector field influences the form of the transverse space, $\tilde{{\bf g}}$, and the metric functions $H^{(2)}$  and $W^{(1)}_i$. Any metric sharing these functions with differing $H^{(1)}$, $H^{(0)}$  and $W^{(0)}_i$ will admit the same $IPD$ vector field. 
%xxyyzz

From this result, we have demonstrated that a transitive set of nil-Killing vector fields exist for any Kundt-$CSI$ spacetime. This was achieved using a mapping from an arbitrary Kundt-$CSI$ spacetime to a unique Kundt-$CSI$ spacetime of alignment type ${\bf D}^k$ with identical $SPIs$, $\mathcal{I}$, as the original spacetime but whose Cartan invariants are characterized by the set $\mathcal{I}$. Such spacetimes admit a transitive set of Killing vector fields, i.e., they are locally homogeneous, and Theorem \ref{thm:IPDKundt} implies that these vector fields are nil-Killing $IPD$ vector fields for the original Kundt-$CSI$ spacetime. 

Admittedly, this mapping will not work for other $CSI$ pseudo-Riemannian spaces of indefinite signature, as it may yield  spaces which are not locally homogeneous but are $CSI$. We hope the frame approach introduced here will be helpful in identifying $IPD$-vector fields for pseudo-Riemannian spaces. However, this task is complicated by the higher dimensional b.w. structure of the pseudo-Riemannian spaces. To illustrate the issue, we will consider an $(2k+m)$-dimensional manifold of signature $(k, k+m)$, a null coframe can be chosen such that 
\beq ds^2 = 2(\bl^1 {\bf n}^1 + \bl^2 {\bf n}^2 + \hdots + \bl^k {\bf n}^k) + \delta_{ab} {\bf m}^a {\bf m}^b,~~ a,b \in [1, m]. \label{badnull} \eeq
\noindent Relative to this coframe, the Abelian subgroup of the group $SO(k, k+m)$ are boosts in each of the $k$ null planes: 
\beq (\bl^i, {\bf n}^i) \to (e^{\lambda_i} \bl_i, e^{-\lambda_i} {\bf n}_i) \label{badboost} \eeq
\noindent for $i \in [1,k]$ where $\lambda_i$ are real-valued. 
In analogy with the Lorentzian case, we have the concept of {\it boost weights} ${\bf b} \in \mathbb{Z}^k$ such that for an arbitrary component of a rank $n$ tensor ${\bf T}$ with respect to the coframe \eqref{badnull}, a boost in each of the $k$ null planes gives the transformation 
\beq T_{\mu_1 \hdots \mu_n} \to e^{(b_1 \lambda_1 + b_2 \lambda_2 + \hdots + b_k \lambda_k)} T_{\mu_1 \hdots \mu_n}, \nonumber \eeq
\noindent where $b_1, \hdots, b_k$ are integers and ${\bf b} = (b_1, b_2, \hdots, b_k)$ is the boost weight vector of the component $T_{\mu_1 \hdots \mu_n}$. We can decompose the tensor ${\bf T}$ into the following decomposition
\beq {\bf T} = \sum\limits_{{\bf b} \in \mathbb{Z}^k} (T)_{{\bf b}}. \eeq
\noindent Here, $(T)_{{\bf b}}$ denotes the projection onto the subspace of components of boost weight {\bf b}. 

With the boost weight decomposition, we can introduce properties to classify tensors in a similar manner to the alignment classification: 
\begin{defn} \label{defn:Sprop}
Consider the conditions 
\beq \begin{aligned} B1)~ & (T)_{{\bf b}} =0, \text{ for all } {\bf b} = (b_1, b_2, \hdots, b_k), b_1 >0, \\
B2)~ & (T)_{{\bf b}} =0, \text{ for all } {\bf b} = (0, b_2, \hdots, b_k), b_2 >0, \\ 
& \vdots \\
Bk)~ & (T)_{{\bf b}} =0, \text{ for all } {\bf b} = (0, 0, \hdots, 0, b_k), b_k >0, \end{aligned} \eeq
\noindent A tensor ${\bf T}$ possesses the ${\bf S}_i$ property, $i \in [1,k]$, if there exists a null coframe such that the conditions $B1)-Bi)$ holds. 
\end{defn}

\begin{defn} \label{defn:Nprop}
A tensor ${\bf T}$ posses the ${\bf N}$ property if a null coframe exists such that $B1)-Bk)$ are satisfied and $$ (T)_{{\bf b}} =0, \text{ for all } {\bf b} = (0,0, \hdots 0,0).$$
\end{defn}

For indefinite signatures other than Lorentzian signature there is another set of properties that must be considered. A tensor which does not have the $S_i$ property can still have a degenerate structure, since the boost weights are a lattice ${\bf b} \in \mathbb{Z}^k \subset \mathbb{R}^k$, we can use a linear transformation ${\bf G} \in GL(k)$ to map the boost weight onto a lattice $\Gamma$ in $\mathbb{R}^k$. If such a  ${\bf G} \in GL(k)$ exists such that the image of the boost weights ${\bf G} {\bf b}$ of the tensor {\bf T} now satisfies some of the properties above, we say the tensor {\bf T} possesses the ${\bf S}_i^{\bf G}$ or ${\bf N}^{\bf G}$ property.

Any tensor satisfying {\it at least} the $S_1^{\bf G}$ property will not be characterized by its invariants \cite{SHY2015}. Thus, a given pseudo-Riemannian space is $\mathcal{I}$-degenerate if the curvature tensor and its covariant derivatives satisfy the $S_1^{\bf G}$ property relative to a common null coframe. As in the Lorentzian case \cite{SM2018}, the proof of this result relies on the limit of a diffeomorphism associated with an appropriately chosen boost in order to generate a non-diffeomorphic space with the same set $\mathcal{I}$. Motivated by this result, and theorem \ref{thm:LorExist} we can state a simple existence theorem for $IPD$ vector fields in pseudo-Riemannian spaces:

\begin{thm} \label{thm:PRexist}
Consider a pseudo-Riemannian space, for which the curvature tensor and its covariant derivatives satisfies the $S_1^{\bf G}$  property with ${\bf G} \in GL(k)$ relative to a fixed coframe basis. Denoting $\mathcal{R}_{\bf G}$ as the ${\bf G}$-transformed components of the curvature tensor and its covariant derivatives, if 
\beq 0 \leq rank([\mathcal{R}^q_{\bf G}]_{b_i = 0}) < rank(\mathcal{R}^q_{\bf G}), \nonumber \eeq
\noindent then the manifold admits a non-trivial $IPD$ vector field, ${\bf X}$, such that 
\beq \mathcal{L}_{\bf X} \mathcal{I} = 0. \nonumber \eeq
\noindent That is, the pseudo-Riemannian space is $\mathcal{I}$-degenerate.
\end{thm}

In principle the set of $IPD$ vector fields can be determined using this approach. However, in practice this is too difficult to compute for a generic pseudo-Riemannian manifold due to the $S_i^{\bf G}$ property. As an alternative, theorem \ref{thm:PRexist} can be restated in terms of differential invariants \cite{kruglikov2016global} by comparing the rank of $\mathcal{I}$ to the rank of the set of differential invariants. In future work, we will explore alternative approaches to finding a transitive set of  nil-Killing $IPD$ vector fields for spacetimes using the geometric evolution equations \cite{MTA2018} with the goal of extending the approach to pseudo-Riemannian spaces of other signatures. 

\end{section}

\begin{section}*{Acknowledgements}

We would like to thank Sigbj{\o}rn Hervik for helpful discussions during the course of this project. This work was supported through the Research Council of Norway, Toppforsk grant no. 250367: Pseudo-
Riemannian Geometry and Polynomial Curvature Invariants: Classification, Characterisation and Applications.

\end{section}

\bibliographystyle{unsrt-phys}
\bibliography{IPDReferences}

\end{document}